\documentclass[10pt,a4paper,runningheads,envcountsame]{llncs}
\usepackage%[a4paper,margin=2.5cm]
{geometry}
\usepackage{graphicx,amssymb,amsmath}
\usepackage{physics}
\usepackage{enumitem}
\usepackage{tikz}

\usepackage{mathdots}
\usepackage{yhmath}
\usepackage{cancel}
\usepackage{color}
\usepackage{siunitx}
\usepackage{array}
\usepackage{multirow}
\usepackage{gensymb}
\usepackage{tabularx}
\usepackage{extarrows}
\usepackage{booktabs}
\usetikzlibrary{fadings}
\usetikzlibrary{patterns}
\usetikzlibrary{shadows.blur}
\usetikzlibrary{shapes}
\usepackage{xspace}
\usepackage{caption}
\usepackage{subcaption}
\usepackage{amsmath}
\usepackage{refcount}
\usepackage{hyperref}
\hypersetup{hidelinks,colorlinks=true,allcolors=black,pdfstartview=Fit,breaklinks=true}
\input{figs}
\title{Geometric Burning Under \(L_1\) and \(L_\infty\) Metrics, and Beyond}

\author{Shahin Kamali \and Saba Yazdani}

\institute{
Department of Electrical Engineering and Computer Science \\
York University, Toronto, Canada\\
\email{\{kamalis,yazdanis\}@yorku.ca}
}

\authorrunning{S. Kamali and S. Yazdani}

\makeatletter

\spnewtheorem{observation}[theorem]{Observation}{\bfseries}{\itshape}

\makeatother

\date{}

\begin{document}
\thispagestyle{empty}
\maketitle
\date{}

\begin{abstract}
Burning is a discrete-time model for propagation in which a new fire starts in each round, while each existing fire expands by one unit of distance along the underlying metric. In geometric burning, the input is a finite point set, and the goal is to burn all points in as few rounds as possible. Equivalently, burning a point set in $k$ rounds corresponds to covering it with metric balls of distinct radii in $\{0,1,\ldots,k-1\}$; the objective is to minimize $k$. Previous work has studied the problem mainly under the Euclidean metric. In this paper, we study geometric burning under the $L_1$ and $L_\infty$
metrics. The problem remains NP-hard in both settings.

The $L_1$ and $L_\infty$ metrics provide additional geometric structure, which allows us to obtain improved approximation guarantees, especially for anywhere burning. We first present a simple $(2+\varepsilon)$-approximation for both anywhere burning and point burning. We then improve the anywhere burning approximation to $7/4+\varepsilon=1.75+\varepsilon$, and give a
$(3151/1620+\varepsilon)$-approximation for point burning, where
$3151/1620<1.9451$. We also extend the anywhere burning result under
$L_\infty$ to every fixed dimension $d\ge 3$ to achieve a $\left(2-\frac{1}{2^{d+1}}+\varepsilon\right)$-approximation. Finally, using standard comparisons between planar $L_p$
distances, we transfer our $L_1$ and $L_\infty$ algorithms, together with
known Euclidean burning algorithms, to obtain approximation guarantees for
every fixed $1\le p\le\infty$.

\end{abstract}

% \begin{abstract}
% Burning is a discrete-time model for propagation, in which a new fire starts in each round, while existing fires continue to expand. In geometric burning, the input is a finite point set, and the goal is to burn all points in a minimum number of rounds. Equivalently, burning a point set in $k$ rounds corresponds to covering it with metric balls of distinct radii in $\{0,1,\ldots,k-1\}$. Previous work has studied the problem mainly under the Euclidean metric. In this paper, we study geometric burning under the $L_1$ and $L_\infty$ metrics. 
% The problem remains NP-hard in both settings. These metrics provide additional geometric structure, which allows us to obtain improved approximation guarantees, especially for anywhere burning. We first present a simple $(2+\varepsilon)$-approximation for both anywhere burning and point burning, then improve the anywhere burning approximation to $7/4+\varepsilon=1.75+\varepsilon$, and also give a $(3151/1620+\varepsilon)$-approximation for point burning, where $3151/1620<1.9451$. We also extend the anywhere burning result to fixed dimension $d$, obtaining a $(2-2^{-d}+\varepsilon)$-approximation.
% \keywords{Geometric burning, approximation algorithms, \(L_1\) metric, \(L_\infty\) metric, \(L_p\) metrics, distinct-radii covering, geometric set cover, point burning, anywhere burning}
% \end{abstract}

%\scalebox{.8}{\invisbox}
%\newpage
%\setcounter{page}{1}

\section{Introduction}
Burning is a discrete-time spreading process in which new \emph{fires} are activated over time while older fires continue to expand. The process was first studied on graphs: given an undirected, connected graph $G$, at each round, one new vertex is chosen as the source of a new fire, while old fire spreads from each burned vertex to its neighbors. The goal is to burn the entire graph as quickly as possible~\cite{BonatoSuv21}.

The same idea can be formulated in any metric space. Let $(X,d)$ be a metric space, and let $\mathcal{P}\subseteq X$ be a finite set of points. In each round, a new source is selected, while all existing fires expand by one unit according to the metric $d$; that is, any point $p\in\mathcal{P}$ that lies within distance $1$ of a burned point becomes burned in the current round. In graph burning, $X$ is the vertex set of a graph and $d$ is the shortest-path metric. In geometric burning, $X$ is typically $\mathbb{R}^2$ equipped with a geometric metric, and $\mathcal{P}\subset X$ is the finite input point set.

We study geometric burning under the Manhattan $(L_1)$ and Chebyshev $(L_\infty)$ metrics. For two points $p=(x_p,y_p)$ and $q=(x_q,y_q)$, these distances are defined by $d_1(p,q)=|x_p-x_q|+|y_p-y_q|$ and $d_\infty(p,q)=\max\{|x_p-x_q|,|y_p-y_q|\}$, respectively. Under $L_1$, metric balls are diamonds, while under $L_\infty$, metric balls are axis-aligned squares. These two settings are closely related, and the main geometric arguments can be stated for axis-aligned squares and then transferred to the $L_1$ metric.

We consider two variants. In the \emph{anywhere burning} problem, sources may be placed anywhere in the ambient space. In the \emph{point burning} problem, sources must be selected from the input set $\mathcal{P}$. That is, point burning restricts sources to input points, while anywhere burning allows arbitrary centers.

A burning sequence of length $k$ can be viewed as a covering of points by metric balls: the source chosen in round $i$ has radius $k-i$ by the end of the process. Thus, the problem is equivalent to covering $\mathcal{P}$ with balls of distinct radii in $\{0,1,\ldots,k-1\}$, where centers are arbitrary in the anywhere variant and restricted to $\mathcal{P}$ in the point variant (See Figure~\ref{fig:pointbb}). %Figure~\ref{fig:pointbb} illustrates the anywhere and point burning models.
For Euclidean burning, the covering objects are disks. In contrast, the $L_1$ and $L_\infty$ metrics lead to diamonds and squares. This difference is useful algorithmically. In particular, the square structure allows us to use smaller-radius balls more efficiently. % than the standard approach, which first covers the input with large congruent balls.

\subsection{Previous Work}
Burning was first introduced on graphs as a model for the spread of influence through a network~\cite{burn2}. Since then, the graph-burning problem has been studied for several graph classes from both complexity and approximation perspectives. It is NP-hard even for restricted graph families such as trees and disjoint unions of paths~\cite{BessyBJRR17}. For general graphs, a $3$-approximation algorithm is known~\cite{BessyBJRR17,graph1}, and %this bound can be slightly improved to $3-1/b(G)$, where $b(G)$ denotes the burning number of the input graph~\cite{garcia2022burning}. P
polynomial-time approximation schemes (PTASs) are known for disjoint path forests~\cite{graph1} %, trees, 
and graphs of bounded treewidth~\cite{lieskovsky2022graph}. Other graph families, including interval graphs, have also been considered~\cite{graph3}. We refer to the survey of Bonato~\cite{BonatoSuv21} for a broader overview of graph-burning results.

Geometric variants of graph burning have also been studied. % in other settings. For geometric graph classes, 
Gorain et al. gave a $2$-approximation algorithm for square grids and proved NP-completeness for connected interval graphs~\cite{new4}. Bonato et al. considered burning on growing grid graphs~\cite{pla1}, while Evans and Lin studied polygon burning and presented a $3$-approximation algorithm for polygonal domains~\cite{evans2022polygon}.

%For finite point sets in the Euclidean plane, 
Keil et al. introduced the anywhere and point burning problems %for finite point sets 
in the Euclidean plane and proved them to be NP-hard~\cite{keil2022burning}. They also gave $(2+\varepsilon)$-approximation algorithms for both variants~\cite{keil2022burning}. Gokhale, Keil, and Mondal later improved these approximation factors to $1.92188+\varepsilon$ for anywhere burning and $1.96296+\varepsilon$ for point burning~\cite{gokhale2023improved}. Kamali and Shabanijou further improved the Euclidean bounds to $1.8\bar{33}+\varepsilon$ for anywhere burning and $1.944+\varepsilon$ for point burning~\cite{kamali2023improved}.

% \begin{figure}[!t]
% \begin{subfigure}[b]{.24\textwidth}\centering
% \hspace*{-.1cm}\scalebox{1.4}{\anywhereOne}
% \caption{Anywhere burning\\under $L_1$}\label{fig:anywhere-l1}\end{subfigure}%
% \begin{subfigure}[b]{.23\textwidth}\centering
% \hspace*{-.4cm}\scalebox{1.4}{\pointOne}\caption{Point burning\\under $L_1$}\label{fig:point-l1}\end{subfigure}%
% \hfill
% \begin{subfigure}[b]{.24\textwidth}\centering
% \hspace*{-.05cm}\scalebox{1.4}{\anywhereInf}\caption{Anywhere burning\\under $L_\infty$}\label{fig:anywhere-linf}\end{subfigure}%
% \begin{subfigure}[b]{.23\textwidth}\centering
% \scalebox{1.4}{\pointInf}\caption{Point burning\\under $L_\infty$}\label{fig:point-linf}\end{subfigure}
% \caption{An illustration of the burning protocols. Here, $c_i$ indicates the $i$th source in the burning sequence. In this example, anywhere burning and point burning take $5$ and $6$ rounds, respectively.}\label{fig:point}
% \end{figure}
\begin{figure}[!t]
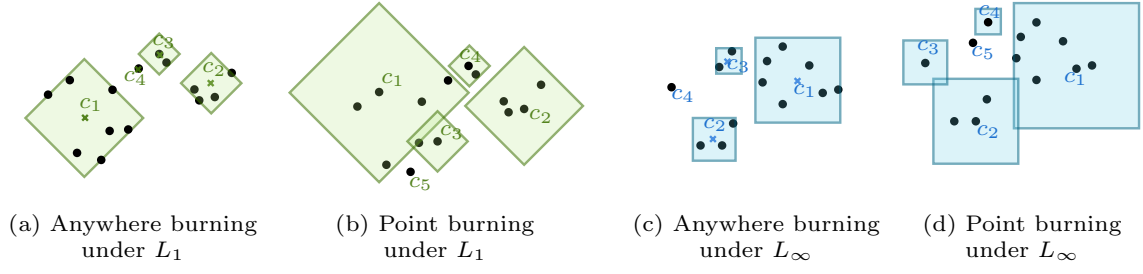

\centering
\resizebox{\textwidth}{!}{\burningOverview}
\caption{An illustration of the burning protocols. Here, $c_i$ indicates the $i$th source in the burning sequence. In this example, anywhere burning and point burning take $4$ and $5$ rounds, respectively.}\label{fig:pointbb}
\end{figure}

\subsection{Contribution}
This paper studies geometric burning under the $L_1$ and $L_\infty$ metrics. %The hardness proof of Keil et al.~\cite{keil2022burning} already applies to one-dimensional instances, where input points are collinear. %Since such instances can be embedded in the plane under both the $L_1$ and $L_\infty$ metrics, the same reduction implies NP-hardness in both settings.
Since the hardness construction of Keil et al.~\cite{keil2022burning} already applies to collinear point sets, it immediately implies that geometric burning is NP-hard under both the $L_1$ and $L_\infty$ metrics.
% \begin{proposition}\label{prop:hardness}
% The geometric burning problem is NP-hard under both $L_1$ and $L_\infty$ metrics.
% \end{proposition}
We therefore focus on approximation algorithms. Our results are as follows.

\begin{itemize}
\item We formulate planar burning under \(L_1\) and \(L_\infty\) as a
distinct-radii square-cover problem; see Definition~\ref{def:drsc}. Under
\(L_\infty\), metric balls are axis-aligned squares, and under \(L_1\), the transformation \(\phi(x,y)=(x+y,x-y)\) maps balls to axis-aligned squares of the same radius. Thus, it suffices to work with square covers throughout the main planar part of the paper.

\item \textbf{Basic Approximation.} We present a basic \((2+\varepsilon)\)-approximation algorithm for both anywhere burning and point burning under \(L_\infty\) and \(L_1\); see Theorem~\ref{thm:two-approx}. For anywhere burning, where square centers are unrestricted, we use the PTAS of Hochbaum and Maass~\cite{HochbaumMaass85} for covering points by congruent axis-aligned squares. For point burning, where centers are restricted to \(\mathcal{P}\), we use the PTAS of Mustafa and Ray~\cite{mustafa2010improved} for the corresponding planar discrete-covering problem.

\item \textbf{Anywhere burning in the plane.} We improve the anywhere-burning approximation factor to
\(7/4+\varepsilon\) in the plane; see Theorem~\ref{thm:seven-fourths}. The improvement comes from using smaller-radius squares that are ignored by the basic approach. In particular, we show that the squares of radii in \(\{1,2,\ldots,g\}\) can be used to cover \(g/4\) squares of radius \(g\), and that this is best possible within this replacement framework.

\item \textbf{Higher-dimensional anywhere burning.} We extend the anywhere-burning result under \(L_\infty\) to fixed higher dimensions; see Theorem~\ref{thm:higher-d-anywhere}. For every fixed \(d\geq 3\), using the Hochbaum--Maass PTAS for covering points by congruent axis-aligned \(d\)-cubes in \(\mathbb{R}^d\), we achieve a
$\left(2-{1}/{2^{d+1}}+\varepsilon\right)$-approximation algorithm for anywhere burning under the \(L_\infty\) metric.

\item \textbf{Point burning in the plane.} For point burning in the plane, we improve the basic approximation by combining the planar discrete-covering PTAS of Mustafa and Ray with a point-centered replacement scheme; see Theorem~\ref{thm:point-burning}. This gives a $\left({3151}/{1620}+\varepsilon\right)$-approximation algorithm for point burning under both \(L_\infty\) and \(L_1\),
where \(3151/1620<1.9451\).

\item \textbf{Transfer to $L_p$ metrics.} Finally, we prove a black-box metric-transfer lemma for planar
\(L_p\) metrics; see Lemma~\ref{lem:norm-transfer}. The lemma shows that any
\((\alpha_q+\varepsilon)\)-approximation algorithm under \(L_q\) can be
converted into an $\left(\alpha_q\,2^{|1/p-1/q|}+\varepsilon\right)$-approximation algorithm under \(L_p\), for every fixed
\(1\le p,q\le\infty\). Applying this lemma with \(q\in\{1,2,\infty\}\), and combining our \(L_1\) and \(L_\infty\) algorithms with known Euclidean algorithms, gives approximation guarantees for both anywhere burning and point burning under every fixed planar \(L_p\) metric; see
Theorems~\ref{thm:lp-transfer} and~\ref{thm:lp-transfer-point}. These bounds interpolate between our square-based guarantees for \(L_1\) and
\(L_\infty\) and the known disk-based guarantees for \(L_2\).
\end{itemize}

Some proofs and routine details are deferred to the appendix due to space
constraints.

\section{Preliminaries}\label{sec:preliminaries}
\paragraph{Problem Definition.}
For a point $c$ in a metric space $(X,d)$ and a radius $r\geq 0$, let $B_d(c,r)=\{x\in X : d(x,c)\leq r\}$ denote the ball of radius $r$ centered at $c$.

Let $\mathcal{P}$ be a finite set of points in a metric space $(X,d)$. A burning sequence of length $k$ is a sequence of sources $c_1,\ldots,c_k$ such that $\mathcal{P}\subseteq \bigcup_{i=1}^k B_d(c_i,k-i)$. Thus, a burning sequence of length $k$ is equivalently a cover of $\mathcal{P}$ by balls with distinct radii in $\{0,1,\ldots,k-1\}$. In the \emph{anywhere burning} variant, the sources may be chosen from the ambient space $X$; in the \emph{point burning} variant, they must belong to $\mathcal{P}$.

We focus on the metrics $L_\infty$ and $L_1$ on $\mathbb{R}^2$, defined by
$d_\infty(p,q)=\max\{|x_p-x_q|,|y_p-y_q|\}$ and
$d_1(p,q)=|x_p-x_q|+|y_p-y_q|$.
When the metric is an $L_p$ metric $d_p$, we write $B_p(c,r)$ as shorthand for
$B_{d_p}(c,r)$. We use the following square-cover formulation.

\begin{definition}[Distinct-Radii Square Cover]\label{def:drsc}
Let $\mathcal{P}$ be a set of points in $\mathbb{R}^2$. A \emph{distinct-radii square cover} of length $k$ is a set of axis-aligned squares $B_\infty(c_0,0),B_\infty(c_1,1),\ldots,B_\infty(c_{k-1},k-1)$ such that $\mathcal{P}\subseteq \bigcup_{i=0}^{k-1} B_\infty(c_i,i)$. In \emph{anywhere covering}, the centers $c_i$ may be arbitrary points of $\mathbb{R}^2$. In \emph{point covering}, the centers must belong to $\mathcal{P}$. The objective is to find the minimum value of $k$ for which such a cover exists.
\end{definition}

The distinct-radii square-cover problem is exactly the covering formulation of burning under the $L_\infty$ metric: a source chosen in round $j$ of a $k$-round burning sequence has radius $k-j$ by the end of the process; after reindexing the sources by their final radii, we obtain squares of radii $0,1,\ldots,k-1$. Similarly, any distinct-radii square cover of length $k$ gives a burning sequence of length $k$ by activating the center of the radius-$(k-j)$ square in round $j$.
The same definition also captures burning under the $L_1$ metric. The transformation $\phi(x,y)=(x+y,x-y)$ satisfies $d_\infty(\phi(p),\phi(q))=d_1(p,q)$ for all $p,q\in\mathbb{R}^2$. Hence, applying $\phi$ maps $L_1$-balls to axis-aligned squares of the same radius. Therefore, an $L_1$ burning instance on $\mathcal{P}$ is equivalent to a distinct-radii square-cover instance on $\phi(\mathcal{P})$.

For this reason, in the rest of the paper, we work with the distinct-radii square-cover formulation. The results for $L_\infty$ follow directly, and the corresponding results for $L_1$ follow by applying the transformation $\phi$.

\subsection{Square-Covering Subroutines}\label{subsec:cover-subroutines}
Throughout the paper, %an axis-aligned square is a closed square region whose sides are parallel to the coordinate axes. We 
we specify an axis-aligned square by its $L_\infty$-radius: a square of radius $r$ centered at $c$ is the set $B_\infty(c,r)$ and has side length $2r$.

\begin{definition}[Congruent Square Cover]\label{def:csc}
An instance of the \emph{Congruent Square Cover} problem consists of a finite point set $\mathcal{P}\subset\mathbb{R}^2$ and a fixed radius $r>0$. The objective is to cover $\mathcal{P}$ with a minimum number of radius-$r$ axis-aligned squares whose centers may be arbitrary points of the plane.
\end{definition}

We use the following result of Hochbaum and Maass for the anywhere version. Their result applies more generally to covering points in fixed dimensions: % by rectilinear blocks of prescribed side lengths.

\begin{theorem}[Hochbaum and Maass~\cite{HochbaumMaass85}]
\label{thm:hochbaum-maass}
For every fixed dimension $d$ and every fixed $\varepsilon>0$, there is a polynomial-time algorithm that, given a finite point set in $\mathbb{R}^d$ and a fixed radius $r>0$, returns a cover by axis-aligned $d$-cubes of radius $r$ whose size is at most $(1+\varepsilon)$ times optimum. In particular, for $d=2$, there is a PTAS for covering points by congruent %axis-aligned 
squares.
\end{theorem}

For the point version, centers are not arbitrary. We use the following discrete-center covering problem.

\begin{definition}[Discrete Congruent Square Cover]\label{def:dcsc}
An instance of the \emph{Discrete Congruent Square Cover} problem consists of a finite point set $\mathcal{P}\subset\mathbb{R}^2$, a finite set $C\subset\mathbb{R}^2$ of candidate centers, and a fixed radius $r>0$. The objective is to find a minimum-size subset $C'\subseteq C$ such that
$\mathcal{P}\subseteq \bigcup_{c\in C'} B_\infty(c,r)$.
That is, we want to cover all points of $\mathcal{P}$ using the minimum number of radius-$r$ axis-aligned squares centered at points of $C$.
\end{definition}

% We will also use the same problem with squares of an arbitrary fixed radius $r>0$. By scaling the plane by a factor of $1/r$, this version is equivalent to the unit-radius version. Thus, an approximation algorithm for unit squares also applies to congruent squares of any fixed radius.

When $r=1$, this problem is usually called \emph{Discrete Unit Square Cover}. By scaling the plane by a factor of $1/r$, the fixed-radius version is equivalent to the unit-radius version. %Thus, any PTAS for Discrete Unit Square Cover also applies to Discrete Congruent Square Cover.
Discrete Congruent Square Cover is a planar geometric set-cover problem with candidate centers and translates of a fixed square. The local-search PTAS of Mustafa and Ray~\cite{mustafa2010improved} applies to this setting and we use it as a black box in our %. We use the following consequence as a black box in the 
point burning algorithms.

\begin{theorem}[Mustafa and Ray~\cite{mustafa2010improved}]
\label{thm:mustafa-ray}
For every fixed $\varepsilon>0$, there is a polynomial-time algorithm for Discrete Congruent Square Cover in the plane that returns a cover with size at most $(1+\varepsilon)$ times optimum.
\end{theorem}

We use the two PTAS results in different settings. The Hochbaum--Maass PTAS applies directly to congruent axis-aligned cubes with unrestricted centers and extends to every fixed dimension; we use it for anywhere burning. The Mustafa--Ray PTAS is used for the planar discrete-center setting, where centers are restricted to a finite candidate set; this is the subproblem that arises in point burning.

\section{Basic $(2+\varepsilon)$-Approximation Algorithms}\label{sec:baseline}
In this section, we give a basic approximation algorithm for the distinct-radii square-cover problem. As discussed in Section~\ref{sec:preliminaries}, this gives the corresponding result for burning under $L_\infty$ and $L_1$. %We present the argument simultaneously for the anywhere and point versions.

Fix $\varepsilon>0$. For each integer guess $g\in \{1,2,\ldots,|\mathcal{P}|\}$, we solve the appropriate congruent-square covering subproblem with radius $g$. For the anywhere version, we apply the PTAS of Hochbaum and Maass (Theorem~\ref{thm:hochbaum-maass}). For the point version, we apply the PTAS of Mustafa and Ray (Theorem~\ref{thm:mustafa-ray}) to the discrete instance with candidate centers $C=\mathcal{P}$. In either case, let $U_g$ be the cover returned by the PTAS, and let $\tau(g)$ denote the minimum number of radius-$g$ squares for the corresponding covering subproblem. Then $|U_g|\leq (1+\varepsilon)\tau(g)$.

We stop at the first (smallest) value of $g$ such that $|U_g|/(1+\varepsilon)\leq g$. Such a value always exists, since for $g=|\mathcal{P}|$ one can cover $\mathcal{P}$ using at most $|\mathcal{P}|$ squares, one centered at each input point. Let $g^*$ be the first successful guess, and let $m=|U_{g^*}|$. By the stopping condition, $m\leq (1+\varepsilon)g^*$.

We now convert $U_{g^*}$ into a distinct-radii square cover. Let $\mathcal{C}$ be a cover of length $L=g^*+m$. The available radii in $\mathcal{C}$ are $0,1,\ldots,L-1$. For each square $B_\infty(c,g^*)$ in $U_{g^*}$, we include in $\mathcal{C}$ a square centered at the same point $c$, using one of the distinct radii $g^*,g^*+1,\ldots,g^*+m-1$. Since each assigned radius is at least $g^*$, the corresponding square in $\mathcal{C}$ covers the assigned radius-$g^*$ square in $U_{g^*}$. The remaining radii in $\mathcal{C}$ (with radii $<g^*$) play no role in the cover. Hence, $\mathcal{C}$ covers $\mathcal{P}$ and has length $L=g^*+m\leq (2+\varepsilon)g^*$.

It remains to relate $g^*$ to the optimum. Since $g^*$ is the first successful guess, the previous guess $g^*-1$ did not satisfy the stopping condition. Hence $|U_{g^*-1}|/(1+\varepsilon)>g^*-1$. %Since $|U_{g^*-1}|\leq (1+\varepsilon)\tau(g^*-1)$, i
It follows that 
$\tau(g^*-1)>g^*-1$. If the minimum distinct-radii square cover had length at most $g^*-1$, then by enlarging all its squares to radius $g^*-1$, we would obtain a cover of $\mathcal{P}$ by at most $g^*-1$ congruent radius-$(g^*-1)$ squares. In the point version, these centers belong to $\mathcal{P}$; in the anywhere version, the centers are unrestricted. This contradicts $\tau(g^*-1)>g^*-1$. Therefore, the optimum is at least $g^*$. %\\ We can conclude the following theorem.

\begin{theorem}\label{thm:two-approx}
For every fixed $\varepsilon>0$, there is a polynomial-time 
$(2+\varepsilon)$-approximation algorithm for both the anywhere and point versions of the distinct-radii square-cover problem. Equivalently, there is a polynomial-time $(2+\varepsilon)$-approximation algorithm for both anywhere burning and point burning under the $L_\infty$ and $L_1$ metrics.
\end{theorem}

\begin{proof}
The argument above gives a cover of length at most $(2+\varepsilon)g^*$, where $g^*$ is a lower bound on the optimum value. This proves the approximation guarantee for the distinct-radii square-cover formulation, and hence for burning under $L_\infty$. The result for $L_1$ follows by applying the transformation $\phi(x,y)=(x+y,x-y)$ from Section~\ref{sec:preliminaries}, solving the transformed instance under $L_\infty$, and mapping the centers back by $\phi^{-1}$. The algorithm considers at most $|\mathcal{P}|$ guesses for $g$, invokes a polynomial-time PTAS for each guess, and performs only polynomial-time operations otherwise; hence, for fixed $\varepsilon$, it runs in polynomial time.
\qed
\end{proof}
\begin{remark}[Basic algorithm in higher dimensions]\label{rem:basic-higher-d}
For anywhere burning, Theorem~\ref{thm:two-approx} extends verbatim to
\(L_\infty\) in every fixed dimension \(d\). One replaces squares by axis-aligned \(d\)-cubes and uses the Hochbaum--Maass PTAS of
Theorem~\ref{thm:hochbaum-maass} for covering points by congruent \(d\)-cubes.
The stopping rule and lower-bound argument from Section~\ref{sec:baseline} are
unchanged, giving a \((2+\varepsilon)\)-approximation for anywhere burning under
\(L_\infty\) in \(\mathbb{R}^d\).
On the other hand, the point version would
require an approximation scheme for the corresponding discrete-center
\(d\)-cube covering problem; see also Remark~\ref{remarkTwo}. Moreover, the
\(L_1\)-to-\(L_\infty\) transformation used in
Section~\ref{sec:preliminaries} is planar.
\end{remark}

\section{A $(7/4+\varepsilon)$-Approximation for Anywhere Burning}\label{sec:anywhere}
In this section, we improve the basic approximation algorithm for the anywhere version. The improvement comes from using the radii smaller than a guessed value $g$, which are ignored by the basic algorithm. We first prove the geometric covering lemma that allows us to use these smaller radii.

\begin{lemma}\label{lem:quarter-cover}
Let $g$ be a positive integer divisible by $4$. Suppose we have one axis-aligned square of each radius in $\{1,2,\ldots,g\}$. The maximum number of radius-$g$ squares that can be covered using these squares, with each available square used at most once, is exactly $g/4$.
\end{lemma}

\begin{proof} \ 
%We first prove the upper bound. 
\textbf{Upper bound:} A square of radius strictly smaller than $g$ cannot cover two distinct corners of a radius-$g$ square. %, since the $L_\infty$-distance between any two distinct corners of such a square is $2g$. 
Therefore, any radius-$g$ square %that is not covered by the single available square of radius $g$ 
requires at least four smaller squares. %, one for each corner. 
So, if $t$ radius-$g$ squares are covered, at most one can be covered by the single square of radius $g$, and the rest each require at least four squares. 
%. If the square of radius $g$ is not used on any copy, then at least $4t$ squares are needed, so $t\leq g/4$. If the square of radius $g$ is used to cover one copy, then the remaining $t-1$ copies still require at least four squares each. 
Consequently, at least $1+4(t-1)=4t-3$ squares are needed. Since only $g$ squares are available and $g$ is divisible by $4$, we have $t\leq g/4$.\vspace{0.5mm}

%We now prove the lower bound. 
\noindent\textbf{Lower bound:}
For each $i\in \{0,1,\ldots,g/4-1\}$, let $j=2i+1$ and form the group $G_i=\{j,j+1,g-j,g-j+1\}$. These $g/4$ groups are pairwise disjoint and together use all radii in $\{1,2,\ldots,g\}$. %Fix one group $G_i$ and write $j=2i+1$. 
%Consider a target square of radius $g$. 
We cover a target square of radius $g$ with any group $G_i$. For that, we place the squares of radii $g-j+1$ and $g-j$ at two opposite corners of the target square, and place the squares of radii $j$ and $j+1$ at the two remaining corners (note that $j=2i+1$). The two larger squares cover the target square except for two corner regions, which are covered by the two smaller squares. See Figure~\ref{fig:four-cover} for an illustration. Hence, each group covers one radius-$g$ square, giving $g/4$ covered squares in total.
\qed
\end{proof}

\begin{figure}[!t]
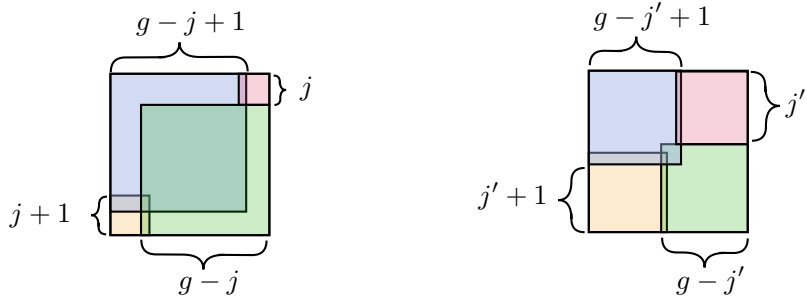

\centering
\scalebox{1.05}{\swastikOne}\ \hspace*{14mm}
\scalebox{1.05}{\swastikTwo}
\caption{An illustration of the covering in Lemma~\ref{lem:quarter-cover} for two different groups. }\label{fig:four-cover}
\end{figure}

We now describe the improved algorithm. Let $\varepsilon>0$ be fixed, and let $\varepsilon'>0$ be chosen sufficiently small. We run the PTAS of Hochbaum--Maass (Theorem~\ref{thm:hochbaum-maass}) for the anywhere version. Let $h$ be the first guess for which the PTAS returns a cover $U$ of $\mathcal{P}$ by radius-$h$ squares satisfying $|U|/(1+\varepsilon')\leq h$. Let $m=|U|$. Then $m\leq (1+\varepsilon')h$, and by the same argument as in Section~\ref{sec:baseline}, the optimum is at least $h$.
Let $g$ be the smallest integer divisible by $4$ such that $g\geq h$. Then $g\leq h+3$. By increasing the radius of every square in $U$ from $h$ to $g$, we obtain a cover of $\mathcal{P}$ by $m$ radius-$g$ squares.

By Lemma~\ref{lem:quarter-cover}, the radii set $\{1,2,\ldots,g\}$ can be partitioned into $g/4$ groups, each of which covers one radius-$g$ square. If $m<g/4$, we use $m$ of these groups to cover all squares in $U$, resulting a cover of length at most $g+1$, which is no worse than the bound below. Otherwise, we use these $g/4$ groups to cover $g/4$ squares of $U$. The remaining $m-g/4$ squares of $U$ are covered one by one using the radii $g+1,g+2,\ldots,g+(m-g/4)$.
In the latter case, 
the largest radius used is $g+(m-g/4)$, and the resulting cover has length at most $L=g+(m-g/4)+1=m+3g/4+1$. Using $m\leq (1+\varepsilon')h$ and $g\leq h+3$, we get $L\leq (1+\varepsilon')h+3(h+3)/4+1=(7/4+\varepsilon')h+13/4$. In the former case, $L\leq g+1\leq h+4$, which also satisfies
$L\leq(7/4+\varepsilon')h+O(1)$.

\begin{theorem}\label{thm:seven-fourths}
For every fixed $\varepsilon>0$, there is a polynomial-time $(7/4+\varepsilon)$-approximation algorithm for the anywhere version of the distinct-radii square-cover problem. Equivalently, there is a polynomial-time $(7/4+\varepsilon)$-approximation algorithm for anywhere burning under the $L_\infty$ and $L_1$ metrics.
\end{theorem}

\begin{proof}
The construction above gives a cover of length at most $(7/4+\varepsilon')h+O(1)$, where $h$ is a lower bound on the optimum. If $h$ is larger than a constant depending only on $\varepsilon$, the additive term is absorbed into the $\varepsilon h$ term by choosing $\varepsilon'$ sufficiently small. %If $h$ is bounded by that constant, we compute an optimal solution by exhaustive search. The result for $L_1$ follows from the transformation $\phi(x,y)=(x+y,x-y)$.
If $h$ is bounded by that constant, then the algorithm gives a solution of constant length $L$, and hence $\mathrm{OPT}\leq L$. We can then find an optimum by exhaustive search: for each $k\in \{1,2,\ldots,L\}$, and for each radius $r\in\{0,1,\ldots,k-1\}$, the possible centers can be discretized to the polynomial-size set
$\{p_x-r,p_x+r:p\in\mathcal{P}\}\times\{p_y-r,p_y+r:p\in\mathcal{P}\}$.
Since $L$ is constant, enumerating all choices and checking coverage takes polynomial time.
\qed\end{proof}

\begin{remark}
\emph{The replacement step in Theorem~\ref{thm:seven-fourths} is best possible in the following sense. After the PTAS returns $g$ congruent radius-$g$ squares, Lemma~\ref{lem:quarter-cover} shows that the smaller radii in $\{1,2,\ldots,g\}$ can cover at most $g/4$ of them. Therefore, any further improvement in the approximation ratio requires using more information about the input point set, or an approach that does not rely on first reducing the instance to a congruent-square cover and then covering those squares independently.} 
\end{remark}

%%%%%%%%%%%%%%%%%%%%%%%%%
\subsection{Extension to Higher Dimensions}\label{subsec:higher-dim}

%The planar case was handled more sharply in Theorem~\ref{thm:seven-fourths}.
Here we give a simple extension of the anywhere-burning algorithm to every fixed dimension \(d\geq 3\) under the \(L_\infty\) metric.
An axis-aligned cube in \(\mathbb{R}^d\) with side length \(2r\) is called a
\(d\)-cube of radius \(r\). We use the Hochbaum--Maass PTAS
(Theorem~\ref{thm:hochbaum-maass}) for covering points in \(\mathbb{R}^d\) by
congruent axis-aligned \(d\)-cubes. %The only additional ingredient is the following replacement lemma.
The only additional ingredient is the following replacement lemma, whose proof is deferred to the appendix. The idea is as follows: partition each radius-\(g\) cube into its \(2^d\) corner subcubes, each of side length \(g\). A group of \(2^d\) available cubes, each of radius at least \(g/2\), can then be placed at the corners to cover one radius-\(g\) cube.

\newcommand{\STATEMENTlemdcubecover}{Let \(d\geq 2\) be fixed, and let \(g\) be a positive integer divisible by \(2^{d+1}\). Suppose we have one axis-aligned \(d\)-cube of each radius in \(\{1,2,\ldots,g-1\}\). Then these cubes can be used to cover
\(g/2^{d+1}\) radius-\(g\) \(d\)-cubes, with each available cube used at most once.}

\begin{lemma}\label{lem:d-cube-cover}%\emph{[Appendix]}
\STATEMENTlemdcubecover
\end{lemma}

\newcommand{\Prooflemdcubecover}{
\begin{proof}
Let \(Q\) be a radius-\(g\) \(d\)-cube. Partition \(Q\) by its coordinate
midplanes into \(2^d\) congruent subcubes, each of side length \(g\). Each subcube is incident to one corner of \(Q\).
We use only the radii $g/2,g/2+1,\ldots,g-1$; there are \(g/2\) such radii. Since \(g\) is divisible by \(2^{d+1}\), these
radii can be partitioned into \(g/2^{d+1}\) groups, each of size \(2^d\).
Consider one such group. Assign its \(2^d\) radii arbitrarily to the \(2^d\) corners of \(Q\), and place each cube at the corresponding corner. Since every
assigned radius is at least \(g/2\), each such cube has side length at least
\(g\), and therefore covers the corresponding subcube in the midpoint
partition of \(Q\). Hence, the \(2^d\) cubes in the group cover all of \(Q\).
Applying this construction independently to each group gives covers for
\(g/2^{d+1}\) radius-\(g\) \(d\)-cubes, using each available radius at most
once.
\qed
\end{proof}}
%\Prooflemdcubecover

We now describe the algorithm. Fix \(\varepsilon>0\), and let
\(\varepsilon'>0\) be chosen sufficiently small. For each integer guess \(h\),
we apply the Hochbaum--Maass PTAS of Theorem~\ref{thm:hochbaum-maass} to cover
\(\mathcal{P}\subset\mathbb{R}^d\) by congruent axis-aligned \(d\)-cubes of
radius \(h\). Let \(h^*\) be the first guess for which the returned cover \(U\)
satisfies $|U|/(1+\varepsilon')\leq h^*.$
As in Section~\ref{sec:baseline}, the optimum anywhere-burning value is at
least \(h^*\).

Let \(g\) be the smallest integer divisible by \(2^{d+1}\) such that
\(g\geq h^*\). Then \(g\leq h^*+2^{d+1}-1\). By increasing the radius of every
cube in \(U\) from \(h^*\) to \(g\), we obtain a cover of \(\mathcal{P}\) by
\(|U|\) radius-\(g\) \(d\)-cubes.

Let $t={g}/{2^{d+1}}.$
Now, if \(|U|<t\), then Lemma~\ref{lem:d-cube-cover} can be used to cover all cubes
in \(U\) using only radii smaller than \(g\), giving a cover of length at most
\(g\). %This is already within the claimed bound, up to the additive constant caused by rounding \(g\).
Otherwise, assume \(|U|\geq t\). By Lemma~\ref{lem:d-cube-cover}, the radii
smaller than \(g\) can be used to cover \(t\) cubes of \(U\). The remaining
\(|U|-t\) cubes are covered one by one using radii at least \(g\). Hence the
resulting burning sequence has length at most $g+(|U|-t)
=
|U|+\left(1-{1}/{2^{d+1}}\right)g.$
Since \(|U|\leq (1+\varepsilon')h^*\) and
\(g\leq h^*+2^{d+1}-1\), this length is at most $\left(2-{1}/{2^{d+1}}+\varepsilon'\right)h^*$ plus an additive constant depending only on \(d\). 
The proof of the approximation guarantee stated in the following theorem is
deferred to the appendix; it handles the additive constant and bounded values
of \(h^*\).
%The proof of the resulting approximation guarantee, stated in the following theorem, including the handling of the additive constant and bounded values of \(h^*\), is deferred to the appendix. 

\newcommand{\STATEMENTthmhigherdanywhere}{For every fixed dimension \(d\geq 3\) and every fixed \(\varepsilon>0\), there is a polynomial-time algorithm with an approximation factor of $\left(2-{1}/{2^{d+1}}+\varepsilon\right) $ for anywhere burning under the \(L_\infty\) metric in
\(\mathbb{R}^d\).}

\begin{theorem}\label{thm:higher-d-anywhere}%\emph{[Appendix]}
\STATEMENTthmhigherdanywhere
\end{theorem}

\newcommand{\Proofthmhigherdanywhere}{
\begin{proof}
% The construction above gives a burning sequence of length at most
% \[
% \left(2-\frac{1}{2^{d+1}}+\varepsilon'\right)h^*+O_d(1),
% \]
% where \(h^*\leq \mathrm{OPT}\). Choose \(\varepsilon'>0\) sufficiently small.
% If \(h^*\) is larger than a constant depending only on \(d\) and
% \(\varepsilon\), the additive \(O_d(1)\) term is absorbed into the
% \(\varepsilon h^*\) term, and the returned sequence has length at most
% \[
% \left(2-\frac{1}{2^{d+1}}+\varepsilon\right)\mathrm{OPT}.
% \]

Run the algorithm described in Subsection~\ref{subsec:higher-dim}. Let \(h^*\)
be the first successful guess, let \(U\) be the corresponding cover returned by
the Hochbaum--Maass PTAS, and let \(g\) be the smallest multiple of
\(2^{d+1}\) such that \(g\ge h^*\). As in Section~\ref{sec:baseline}, we have
\(h^*\le \mathrm{OPT}\).

By Lemma~\ref{lem:d-cube-cover}, the radii smaller than \(g\) can cover
\(g/2^{d+1}\) of the radius-\(g\) cubes obtained from \(U\). The remaining
cubes are covered individually using radii of at least \(g\). Therefore, the returned
burning sequence has length at most
$|U|+\left(1-\frac{1}{2^{d+1}}\right)g.$
Since \(|U|\le (1+\varepsilon')h^*\) and
\(g\le h^*+2^{d+1}-1\), this is at most
$\left(2-\frac{1}{2^{d+1}}+\varepsilon'\right)h^*$ plus an additive constant depending only on \(d\).

It remains to handle the case where \(h^*\) is bounded by a constant depending only on \(d\) and \(\varepsilon\). Then the algorithm already gives a solution
of constant length, and hence the optimum is also bounded by a constant. We can
therefore find an optimal burning sequence by exhaustive search. For each fixed
radius \(r\), the possible centers can be discretized to
$C_r=X_1^r\times\cdots\times X_d^r $ and $
X_j^r=\{p_j-r,p_j+r:p\in\mathcal{P}\}.
$ For fixed \(d\), this set is polynomial in \(|\mathcal{P}|\). Since only a constant number of radii are involved, enumerating all center choices and checking
coverage takes polynomial time.

Therefore, the algorithm is polynomial for every fixed \(d\) and fixed
\(\varepsilon\), and achieves the claimed approximation factor.
\qed
\end{proof}
}

\section{A Point Burning Algorithm}\label{sec:point}
In this section, we consider the point version of the distinct-radii square-cover problem. Recall that, in this variant, every square must be centered at an input point of $\mathcal{P}$. We present an algorithm with approximation factor $3151/1620+\varepsilon<1.9451+\varepsilon$.

\paragraph{Overview.}
As in the basic algorithm, we first use the PTAS for Discrete Unit Square Cover (Theorem~\ref{thm:mustafa-ray}) with candidate centers restricted to $\mathcal{P}$. This gives a cover of $\mathcal{P}$ by congruent squares of radius $g$. The basic $2$-approximation would then cover each of these squares separately using radii at least $g$, leaving all radii smaller than $g$ unused.
For point burning, the smaller radii are harder to use than in the anywhere setting, since their centers must be input points. 

We use the following idea. 
Let $S$ be one of the radius-$g$ squares returned by the PTAS, and let its center be $c\in\mathcal{P}$. We place one smaller square centered at $c$ that covers the middle of $S$. The remaining boundary region is partitioned into small square cells. If a cell contains no input point, it can be ignored. If it contains input points, we center a sufficiently large square at one of these points; this square covers the entire cell, and hence all input points in it. 
With carefully selected replacement patterns, we can 
%
%We use three replacement patterns, based on $5\times 5$, $4\times 4$, and $3\times 3$ grids. Together, these patterns allow us to replace 
cover $89g/1620$ of the radius-$g$ squares using only radii smaller than $g$. The remaining squares are covered using radii at least $g$.

\subsection{Point-Centered Replacement Lemma}\label{subsec:point-replacement}

We first show that a fixed fraction of the radius-$g$ squares returned by the discrete covering step can be replaced by smaller point-centered squares.

\begin{lemma}\label{lem:point-replacement-scheme}
Let $g$ be a positive integer divisible by $12960$ and $\mathcal{U}$ be a set of radius-$g$ axis-aligned squares, each centered at a point of $\mathcal{P}$. If $|\mathcal{U}|\geq \frac{89}{1620}g$, then any subset of $\frac{89}{1620}g$ squares from $\mathcal{U}$ can be covered by squares of distinct radii smaller than $g$, all centered at points of $\mathcal{P}$. More precisely, there is a point-centered distinct-radii square cover, using only radii in $\{1,2,\ldots,g-1\}$, that covers all points of $\mathcal{P}$ contained in these $\frac{89}{1620}g$ %radius-$g$ 
squares.
\end{lemma}

\begin{proof}
Consider the intervals $I_1=[2g/5,g/2)$, $I_2=[g/2,2g/3)$, and $I_3=[2g/3,g)$ of radii below $g$. %The interval 
Then $I_1$ contains $g/10$ radii, $I_2$ contains $g/6$ radii, and $I_3$ contains $g/3$ radii. Radii smaller than $2g/5$ are not used.

\begin{figure}[t!]
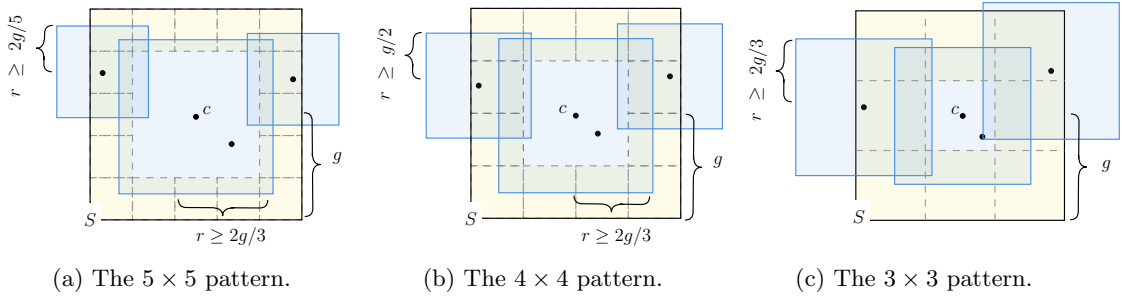

%\centering
\hspace{-.5cm}
\begin{subfigure}[b]{0.32\columnwidth}\centering\scalebox{.58}{\pointFivebyFive}\caption{The $5\times 5$ pattern.}\label{subfig:Five}\end{subfigure}
\begin{subfigure}[b]{0.32\columnwidth}\centering\scalebox{.58}{\pointFourbyFour}\caption{The $4\times 4$ pattern.}\label{subfig:Four}\end{subfigure}
\begin{subfigure}[b]{0.32\columnwidth}\centering\scalebox{.58}{\PointThreebyThreeNew}\caption{The $3\times 3$ pattern.}\label{subfig:Three}\end{subfigure}
\caption{An illustration of the covering in Lemma~\ref{lem:point-replacement-scheme}.}\label{fig:point-patterns}
\end{figure}
\vspace*{2mm} \noindent \textbf{The $5\times 5$ pattern. }
Take $g/160$ of the radius-$g$ squares. For each such square $S$, let $c\in\mathcal{P}$ be its center. Use one radius from $I_3$ to place a square centered at $c$ (See Figure~\ref{subfig:Five}). Since every radius in $I_3$ is at least $2g/3$, this square covers the central $3\times 3$ block of a $5\times 5$ subdivision of $S$ into cells of side length $2g/5$. There are $16$ remaining boundary cells (yellow cells). For each nonempty boundary cell, choose one input point in that cell and center a square of radius at least $2g/5$ at that point. Such a square covers the entire cell, since the $L_\infty$-diameter of the cell is $2g/5$. Therefore, the boundary cells can be covered using at most $16$ radii from $I_1$. Since the pattern is applied to $g/160$ squares, it uses at most $g/10$ radii from $I_1$, exactly the number available. 
\paragraph{The $4\times 4$ pattern.}
Next, take $g/72$ further radius-$g$ squares. For each such square $S$, again place one square centered at the center $c$ of $S$, using a radius from $I_3$ (See Figure~\ref{subfig:Four}). %This covers the middle part of $S$. 
Subdivide $S$ into a $4\times 4$ grid of cells, each of side length $g/2$. The middle $2\times 2$ block is covered by the central square, leaving $12$ boundary cells. For each nonempty cell, choose an input point in the cell and center a square of radius at least $g/2$ at that point. This covers the entire cell. Thus, each radius-$g$ square covered by this pattern uses one radius from $I_3$ and at most $12$ radii from $I_2$. Since the pattern is applied to $g/72$ squares, it uses at most $g/6$ radii from $I_2$, exactly the number available. 

\paragraph{The $3\times 3$ pattern.}
After the first two patterns, the number of unused radii in $I_3$ is $g/3-g/160-g/72$. We use these radii in groups of $9$. For each radius-$g$ square $S$, subdivide $S$ into a $3\times 3$ grid of cells, each of side length $2g/3$. For each nonempty cell, choose an input point in that cell and center a square (of radius at least $2g/3$) at that point to cover %. This covers 
the cell. Hence, each radius-$g$ square can be covered using at most $9$ radii from $I_3$. Hence, the remaining radii in $I_3$ cover $(g/3-g/160-g/72)/9=451g/12960$ radius-$g$ squares.

Combining the three patterns, the total number of radius-$g$ squares replaced using radii smaller than $g$ is $g/160+g/72+451g/12960=89g/1620$.
\qed\end{proof}

\subsection{Algorithm and Analysis}\label{subsec:point-analysis}
Fix $\varepsilon>0$, and let $\varepsilon'>0$ be chosen sufficiently small. We run the PTAS of %Mustafa and Ray for Discrete Unit Square Cover (
Theorem~\ref{thm:mustafa-ray} for Discrete Unit Square Cover with candidate centers %restricted to 
$\mathcal{P}$. For each integer guess $h$, we compute a $(1+\varepsilon')$-approximate cover of $\mathcal{P}$ by radius-$h$ squares centered at points of $\mathcal{P}$. Let $h^*$ be the smallest guess for which the PTAS returns a cover $U$ satisfying $|U|/(1+\varepsilon')\leq h^*$. As in Section~\ref{sec:baseline}, the optimum point burning value is at least $h^*$.

% Let $g$ be the smallest integer divisible by $12960$ such that $g\geq h^*$. Then $g\leq h^*+12959$. By increasing the radius of every square in $U$ from $h^*$ to $g$, we obtain a cover of $\mathcal{P}$ by $m=|U|$ radius-$g$ squares, all centered at points of $\mathcal{P}$.
% %If $m<89g/1620$, then covering all squares in $U$ individually using radii at least $g$ gives a cover of length at most $g+m$, which is already better than the bound below. Otherwise, b
% Applying Lemma~\ref{lem:point-replacement-scheme}, we cover $89g/1620$ of these radius-$g$ squares using only radii smaller than $g$. The remaining $m-89g/1620$ squares are covered one by one using radii at least $g$. Thus, the resulting cover has length at most $g+(m-89g/1620)=m+1531g/1620$. 
% %
% Since $m\leq (1+\varepsilon')h^*$ and $g\leq h^*+12959$, this is at most $(1+\varepsilon')h^*+(1531/1620)(h^*+12959)=(3151/1620+\varepsilon')h^*+O(1)$. For sufficiently large $h^*$, the additive constant is absorbed into the $\varepsilon h^*$ term. If $h^*$ is bounded by a constant depending only on $\varepsilon$, we find an optimal solution by exhaustive search. 

% The proof of the resulting approximation guarantee is deferred to the appendix. The construction above gives a cover of length \((3151/1620+\varepsilon')h^*+O(1)\), where \(h^*\) is a lower bound on the optimum; the remaining details handle the additive constant and bounded values of \(h^*\).

Let \(g\) be the smallest integer divisible by \(12960\) such that \(g\geq h^*\). Then \(g\leq h^*+12959\). By increasing the radius of every square in \(U\) from \(h^*\) to \(g\), we get a cover of \(\mathcal{P}\) by \(m=|U|\) squares of radius \(g\), all centered at points of \(\mathcal{P}\). If \(m<89g/1620\), we cover all squares in \(U\) individually using squares of distinct radii at least \(g\), obtaining a cover of length at most \(g+m< (1709/1620)g\), which is already
within the desired approximation bound up to the same additive constant. Otherwise, by Lemma~\ref{lem:point-replacement-scheme}, we cover \(89g/1620\) of these radius-\(g\) squares using only radii smaller than \(g\). The remaining \(m-89g/1620\) squares are covered one by one using radii at least \(g\). Therefore, the resulting cover has length at most $g+\left(m-{89g}/{1620}\right) = m+{1531g}/{1620}.$ Since \(m\leq (1+\varepsilon')h^*\) and \(g\leq h^*+12959\), this is at most $\left({3151}/{1620}+\varepsilon'\right)h^*+O(1).$ The proof of the resulting approximation guarantee, stated in the following theorem, including the handling of the additive constant and bounded values of \(h^*\), is deferred to the appendix.

\newcommand{\STATEMENTthmpointburning}{For every fixed $\varepsilon>0$, there is a polynomial-time $(3151/1620+\varepsilon)$-approximation algorithm for the point version of the distinct-radii square-cover problem. Equivalently, there is a polynomial-time $(1.9451+\varepsilon)$-approximation algorithm for point burning under both the $L_\infty$ and $L_1$ metrics.}

\begin{theorem}\label{thm:point-burning}%\emph{[Appendix]}
\STATEMENTthmpointburning
\end{theorem}

\newcommand{\Proofthmpointburning}{
\begin{proof}%[of Theorem~\ref{thm:point-burning}]
The described construction returns a point-centered distinct-radii square cover of length at most $(3151/1620+\varepsilon')h^*+O(1)$, where $h^*$ is a lower bound on the optimum point burning value. Choosing $\varepsilon'$ sufficiently small and handling bounded values of $h^*$ by exhaustive search gives the claimed approximation factor for the square-cover formulation, and hence for point burning under $L_\infty$. 
The result for $L_1$ follows from the transformation $\phi(x,y)=(x+y,x-y)$, which maps $L_1$-balls to axis-aligned squares while preserving the relevant distances. Since input points are mapped to input points, the point-center restriction is preserved.
\qed\end{proof}
}
%\Proofthmpointburning

\begin{remark}\label{remarkTwo}\emph{
%The higher-dimensional extension of the point-burning result is less direct when compared to anywhere-burning. 
For point burning, %the covering subproblem for a guessed radius $g$ is a discrete geometric set-cover problem: 
one must choose a minimum subset of the cubes $\{B_\infty(p,g):p\in\mathcal{P}\}$ that covers $\mathcal{P}$. Unlike the anywhere version, this problem is not covered by the PTAS of Hochbaum--Maass. % for freely placed rectilinear blocks.
This distinction is important in \emph{higher dimensions}. For example, in the Euclidean setting, the corresponding discrete covering problem with unit balls in $\mathbb{R}^3$ is APX-hard, and therefore does not admit a PTAS unless
$P=NP$~\cite{ChanGrant14}. For axis-aligned cubes in $\mathbb{R}^3$, similar APX-hardness results are also known~\cite{ChanGrant14}. %On the positive side, constant-factor approximation algorithms are known for covering points by unit cubes in $\mathbb{R}^3$~\cite{ClarksonVaradarajan07}. 
Thus, while higher-dimensional point burning variants may admit constant-factor approximations, we cannot claim a PTAS-based extension of our planar point burning algorithm.}
\end{remark}

\section{A Transfer to Planar \(L_p\) Metrics}
\label{sec:lp-transfer}

The algorithms in Sections~\ref{sec:anywhere} and
\ref{sec:point} use the square structure of the $L_\infty$ metric
and the equivalence between planar $L_1$-balls and axis-aligned squares.
Nevertheless, these algorithms imply approximation guarantees for all planar
\(L_p\) metrics by comparing $L_p$ distances. Moreover, for values of \(p\) close to
\(2\), better bounds are obtained by transferring the Euclidean burning
algorithms of~\cite{kamali2023improved}.

Throughout this section, we assume \(1\le p\le\infty\) and use the convention
\(1/\infty=0\). 
For a vector \(v=(a,b)\in\mathbb{R}^2\), we write
$\|v\|_p=(|a|^p+|b|^p)^{1/p}$ for \(1\le p<\infty\), and $\|v\|_\infty=\max\{|a|,|b|\}$. 
For two points \(u,c\in\mathbb{R}^2\), their \(L_p\)-distance is $\|u-c\|_p,$ 
where \(u-c\) is the displacement vector from \(c\) to \(u\). For a center
\(c\) and a radius \(r\), let $B_p(c,r)=\{u\in\mathbb{R}^2:\|u-c\|_p\le r\}$ denote the \(L_p\)-ball of radius \(r\) centered at \(c\).

The only geometric fact we need is that, in the plane, \(L_p\)-distances and
\(L_q\)-distances differ by at most a constant factor. For every vector
\(v\in\mathbb{R}^2\), if \(p\le q\), then
$\|v\|_q\le \|v\|_p.$ 
If \(p>q\), then the classical comparison inequality for finite-dimensional
\(L_p\) distances gives \[\|v\|_q\le 2^{1/q-1/p}\|v\|_p.\]
This inequality is a standard consequence of Hölder's inequality; see, e.g.,
Hardy, Littlewood, and Pólya~\cite[Chapter~II, Sections~2.7--2.8]{HardyLittlewoodPolya1952}.
The inequality above immediately implies the following observation about how
\(L_p\)-balls and \(L_q\)-balls compare.
\begin{observation}
\label{obs:lp-ball-containment}
For any \(1\le p,q\le\infty\), define $\lambda_{p,q}=2^{\max\{0,1/q-1/p\}}$
and $\mu_{p,q}=2^{\max\{0,1/p-1/q\}}$. 
Then, for every center \(c\) and radius \(r\),
$B_p(c,r)\subseteq B_q(c,\lambda_{p,q}r)$ 
and $B_q(c,r)\subseteq B_p(c,\mu_{p,q}r).$
Moreover, by the definitions of \(\lambda_{p,q}\) and \(\mu_{p,q}\), we have $\lambda_{p,q}\mu_{p,q}=2^{|1/p-1/q|}.$
\end{observation}

Figure~\ref{fig:lp-ball-comparison} illustrates the above observation for several
choices of \(p\) and \(q\).
In what follows, we define $\eta(p,q)=2^{|1/p-1/q|}.$

% fig here (Saba, if you could make it nicer by adjustment):
% https://www.mathcha.io/editor/3Wm5VceMi1zCpmuqNqvpdHlEyge4FGvz5evS9r29Mx
\begin{figure}[!t]
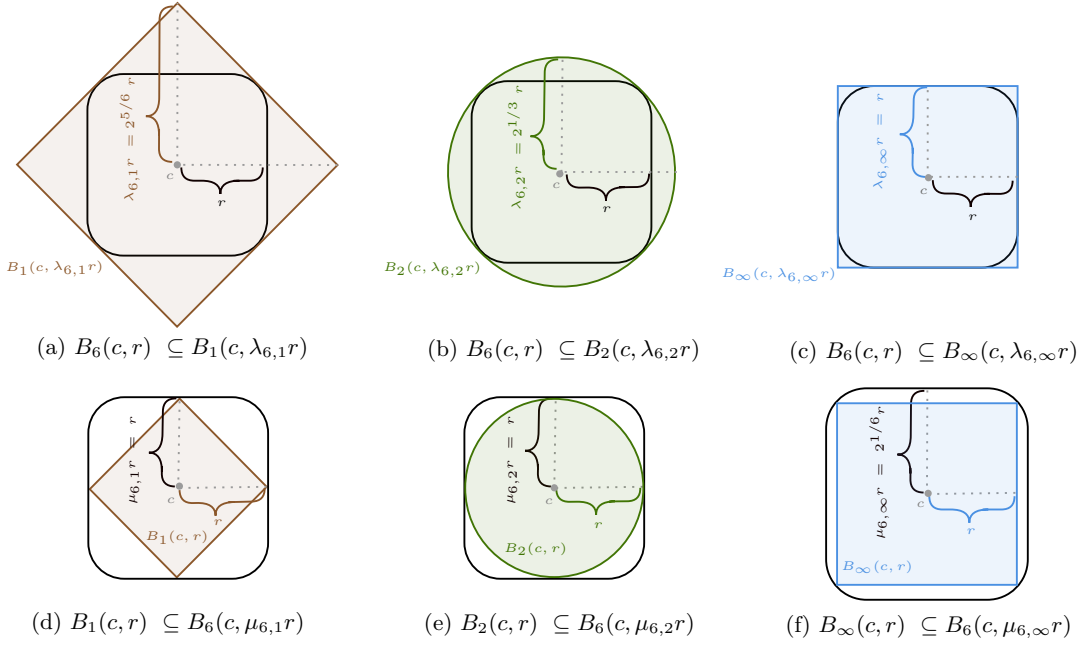

%    \centering
%    \includegraphics[width=0.5\linewidth]{}
\hspace{-.1cm}\scalebox{.95}{\boundFig}
    \caption{An illustration of Observation~\ref{obs:lp-ball-containment} for $p=6$ and $q\in \{1,2, \infty\}.$}
    \label{fig:lp-ball-comparison}
\end{figure}

We now use Observation~\ref{obs:lp-ball-containment} to transfer algorithms
between different planar \(L_p\) metrics. Intuitively, an \(L_p\)-burning
solution can be converted into an \(L_q\)-burning solution by enlarging each
ball by a factor \(\lambda_{p,q}\), and an \(L_q\)-solution can be converted
back to an \(L_p\)-solution by enlarging each ball by a factor \(\mu_{p,q}\).
Thus, the total loss in the approximation factor is
\(\lambda_{p,q}\mu_{p,q}=\eta(p,q)\).

\begin{lemma}[Transfer lemma]
\label{lem:norm-transfer}
Fix \(1\le p,q\le\infty\). Suppose there is a polynomial-time
\((\alpha_q+\varepsilon)\)-approximation algorithm for anywhere burning,
respectively point burning, under the planar \(L_q\) metric. Then, for every
fixed \(\varepsilon>0\), there is a polynomial-time $\left(\alpha_q\eta(p,q)+\varepsilon\right)$
-approximation algorithm for the corresponding burning problem under the
planar \(L_p\) metric.
\end{lemma}

\begin{proof}
We give the proof for one of the two variants; the argument is identical for
anywhere burning and point burning. In the point-burning variant, all
transformations preserve the original centers, so the center restriction is not
affected.

Let \(\operatorname{OPT}_p\) and \(\operatorname{OPT}_q\) denote the optimum
burning numbers under the \(L_p\) and \(L_q\) metrics, respectively. Suppose
first that there is an \(L_p\)-burning sequence of length \(k\). This is a covering of points by \(L_p\)-balls of radii \(0,1,\ldots,k-1\). Using Observation~\ref{obs:lp-ball-containment}, we can cover each ball
\(B_p(c,r)\) by a ball \(B_q(c,\lambda_{p,q}r)\) and rounding radii up to integers to achieve an \(L_q\)-burning sequence of length at most $
\lambda_{p,q}k+O(1).$
Note that, since \(\lambda_{p,q}\ge 1\), the integers
\(\lceil \lambda_{p,q}r\rceil\), for \(r=0,1,\ldots,k-1\), are distinct (so each ball is used at most once, as necessary for the burning protocol).
Therefore, we can write 
\begin{equation}
    \operatorname{OPT}_q \le \lambda_{p,q}\operatorname{OPT}_p+O(1).\label{eq:opt}
\end{equation}

Now run the assumed \(L_q\)-burning algorithm on the same point set. For a
sufficiently small internal error parameter $\varepsilon'$, it returns an \(L_q\)-burning
sequence of length at most $(\alpha_q+\varepsilon')\operatorname{OPT}_q$, that is, a covering of points with $L_q$ balls of radii $0,1, \ldots, (\alpha_q+\varepsilon')\operatorname{OPT}_q$.
By Observation~\ref{obs:lp-ball-containment}, each such \(L_q\)-ball of radius \(r\) is contained in an \(L_p\)-ball of
radius \(\mu_{p,q}r\). Rounding radii up again gives an \(L_p\)-burning
sequence of length at most
\[
\mu_{p,q}(\alpha_q+\varepsilon')\operatorname{OPT}_q+O(1)
\le
\mu_{p,q}(\alpha_q+\varepsilon')
(\lambda_{p,q}\operatorname{OPT}_p+O(1))+O(1).
\]
The first inequality follows from Equation~\ref{eq:opt}. Thus, for sufficiently large \(\operatorname{OPT}_p\), the approximation factor
is at most $\alpha_q\lambda_{p,q}\mu_{p,q}+\varepsilon
=
\alpha_q\eta(p,q)+\varepsilon,$ 
by choosing \(\varepsilon'\) sufficiently small. If \(\operatorname{OPT}_p\) is
bounded by a constant depending only on \(p,q,\varepsilon\), then an optimum
solution can be found by exhaustive search over a constant number of balls.
This is polynomial for fixed \(p,q\). This removes the additive constant and
completes the proof. \qed
\end{proof}

We now apply Lemma~\ref{lem:norm-transfer} with
\(q\in\{1,2,\infty\}\). For \(q=1\) and \(q=\infty\), our algorithms give a
\((7/4+\varepsilon)\)-approximation for anywhere burning and a
\((\rho_1+\varepsilon)\)-approximation for point burning, where $\rho_1=\frac{3151}{1620}.$ 
For \(q=2\), the Euclidean algorithms of~\cite{kamali2023improved} give an
\((11/6+\varepsilon)\)-approximation for anywhere burning and a
\((\rho_2+\varepsilon)\)-approximation for point burning, where
$\rho_2=1.944.$

\begin{theorem}
\label{thm:lp-transfer}
For every fixed \(1\le p\le\infty\) and every fixed \(\varepsilon>0\), there is
a polynomial-time algorithm with an approximation factor of
$
\min\Big\{
\frac{7}{4}\,2^{1/p},
\frac{11}{6}\,2^{|1/p-1/2|},
\frac{7}{4}\,2^{1-1/p}
\Big\}
+\varepsilon$
for anywhere burning under the planar \(L_p\) metric.
\end{theorem}

\begin{proof}
Apply Lemma~\ref{lem:norm-transfer} with \(q=\infty\), \(q=2\), and \(q=1\),
respectively. Since $\eta(p,\infty)=2^{1/p},$ $\eta(p,2)=2^{|1/p-1/2|},$ and $\eta(p,1)=2^{1-1/p}$, the three transferred approximation factors are respectively $\frac{7}{4}2^{1/p},$ $\frac{11}{6}2^{|1/p-1/2|},$ and $\frac{7}{4}2^{1-1/p}.$
Taking the best of the three algorithms gives the claimed bound. \qed
\end{proof}

\begin{theorem}
\label{thm:lp-transfer-point}
For every fixed \(1\le p\le\infty\) and every fixed \(\varepsilon>0\), there is
a polynomial-time ($\rho_p +\varepsilon)$-approximation algorithm for point burning under the planar $L_p$ metric, where
$\rho_p =
\min\Big\{
\rho_1\,2^{1/p},
\rho_2\,2^{|1/p-1/2|},
\rho_1\,2^{1-1/p}
\Big\}
$. Here $\rho_1={3151}/{1620}
\quad\text{and}\quad
\rho_2=1.944.$
\end{theorem}

\begin{proof}
Apply Lemma~\ref{lem:norm-transfer} with \(q=\infty\), \(q=2\), and \(q=1\),
respectively. For \(q=\infty\) and \(q=1\), we use the planar point-burning
algorithm of Theorem~\ref{thm:point-burning}, whose approximation factor is
\(\rho_1+\varepsilon\). For \(q=2\), we use the Euclidean point-burning
algorithm of~\cite{kamali2023improved}, whose approximation factor is
\(\rho_2+\varepsilon\). Since
$\eta(p,\infty)=2^{1/p},$ $\eta(p,2)=2^{|1/p-1/2|},$ and $\eta(p,1)=2^{1-1/p}$, the three transferred approximation factors are respectively $\rho_1 2^{1/p},$ $\rho_2 2^{|1/p-1/2|},$ and $\rho_1 2^{1-1/p}.$
Taking the best of the three algorithms gives the claimed value of \(\rho_p\).
\qed
\end{proof}

\begin{remark}
At the endpoints \(p=1\) and \(p=\infty\), Theorem~\ref{thm:lp-transfer}
recovers the \(7/4+\varepsilon\) anywhere-burning bound proved in
Section~\ref{sec:anywhere}. At \(p=2\), it recovers the
\(11/6+\varepsilon\) Euclidean anywhere-burning bound of~\cite{kamali2023improved}. Thus, the algorithm in this theorem interpolates between the square-based
algorithms for \(L_1,L_\infty\) and the disk-based algorithms for \(L_2\).
\end{remark}

\begin{remark}[The case \(0<p<1\)]
The results of this section are stated for \(1\le p\le\infty\), where the usual
\(L_p\) distance is a metric. For \(0<p<1\), the function \(L_p\) does not define a
metric: if \(e_1=(1,0)\) and \(e_2=(0,1)\), then $\|e_1+e_2\|_p=2^{1/p}>2=\|e_1\|_p+\|e_2\|_p,$
so the triangle inequality fails. This is the standard failure of Minkowski's
inequality below \(p=1\); see, e.g., Hardy, Littlewood, and
Pólya~\cite[Chapter~II]{HardyLittlewoodPolya1952}. Therefore, this range of $p$ does not fit the usual metric interpretation of the burning problem.

% Nevertheless, one may still define a distinct-radii covering problem using the sets
% $B_p(c,r)=\{u\in\mathbb{R}^2:\|u-c\|_p\le r\}.$
% For this covering problem, let $\beta_p=2^{1/p-1}.$
% Then, for every center \(c\) and radius \(r\), we can write $B_p(c,r)\subseteq B_1(c,r)\subseteq B_p(c,\beta_p r).$
% Hence the same proof as Lemma~\ref{lem:norm-transfer}, with \(L_1\) as the
% auxiliary metric, transfers any planar \(L_1\)-based approximation with an
% additional factor \(\beta_p\). In particular, the \(L_1\) algorithms in this
% paper give a
% \[
% \left(\frac{7}{4}\beta_p+\varepsilon\right)
% \]
% -approximation for the anywhere version and a
% \[
% \left(\rho_\square\beta_p+\varepsilon\right)
% \]
% -approximation for the point version of the corresponding \(L_p\)-ball-covering
% problem, where \(\rho_\square=3151/1620\).
\end{remark}

\begin{remark}[Higher dimensions]
The distance-comparison argument is not restricted to the plane. In
\(\mathbb{R}^d\), the same classical comparison inequality for
finite-dimensional \(L_p\) distances gives, for every vector
\(v\in\mathbb{R}^d\),
\[
\|v\|_q
\le
d^{\max\{0,1/q-1/p\}}\|v\|_p.
\]
This inequality is also a standard consequence of Hölder's inequality~\cite{HardyLittlewoodPolya1952}.
Consequently, the analogue of Observation~\ref{obs:lp-ball-containment} holds
in \(\mathbb{R}^d\) after replacing the factor \(2\) by \(d\). Thus,
transferring an \(\alpha_q\)-approximation under \(L_q\) to one under \(L_p\)
incurs a multiplicative loss of $d^{|1/p-1/q|}.$
In particular, Lemma~\ref{lem:norm-transfer} extends to every fixed dimension:
one only needs to replace the planar factor \(2^{|1/p-1/q|}\) by
\(d^{|1/p-1/q|}\).
The concrete bounds in this section are stated only in the plane because they
use planar ingredients: the representation of planar \(L_1\)-balls as
axis-aligned squares, the planar square-covering algorithms used for
\(L_\infty\), and the planar Euclidean disk-burning algorithms
of~\cite{kamali2023improved}.
\end{remark}

\section{Concluding Remarks}

We studied geometric burning under the \(L_1\) and \(L_\infty\) metrics. Using
the equivalence between planar \(L_1\)-balls and axis-aligned squares, we
reduced both settings to a common distinct-radii square-cover formulation. This
square structure yields approximation guarantees beyond the basic
\((2+\varepsilon)\)-framework: a \((7/4+\varepsilon)\)-approximation for
anywhere burning and a \((3151/1620+\varepsilon)\)-approximation for point
burning in the plane.
We also gave a black-box transfer lemma that
converts approximation algorithms between planar \(L_p\) metrics. Together,
these results show that geometric burning benefits substantially from the
structure of the underlying metric, while still admitting uniform guarantees
across the planar \(L_p\) family.

\bibliographystyle{abbrv}
\bibliography{Bibliography}

\newpage

\clearpage
\appendix

\phantomsection
\section{Deferred Proofs}\label{refappxxxxx}

This appendix contains proofs omitted from the main text. For ease of reference, we restate the relevant results before
presenting their proofs.

\setcounter{theorem}{\numexpr\getrefnumber{lem:d-cube-cover}-1\relax}
\begin{lemma}
\STATEMENTlemdcubecover
\end{lemma}
\Prooflemdcubecover

\setcounter{theorem}{\numexpr\getrefnumber{thm:higher-d-anywhere}-1\relax}
\begin{theorem}
\STATEMENTthmhigherdanywhere
\end{theorem}
\Proofthmhigherdanywhere

\setcounter{theorem}{\numexpr\getrefnumber{thm:point-burning}-1\relax}
\begin{theorem}
\STATEMENTthmpointburning
\end{theorem}
\Proofthmpointburning

\end{document}